%% file: root.tex
\documentclass[letterpaper, 10 pt, conference]{ieeeconf}

\IEEEoverridecommandlockouts
\usepackage{cite}
\usepackage{amsmath,amssymb,amsfonts}
\usepackage{algorithmic}
\usepackage{graphicx}
\usepackage{textcomp}
\usepackage{bbm}
\usepackage{color}

\def\BibTeX{{\rm B\kern-.05em{\sc i\kern-.025em b}\kern-.08em
		T\kern-.1667em\lower.7ex\hbox{E}\kern-.125emX}}

\usepackage{enumitem}

\newcommand{\cN}{\mathcal{N}}

\newcommand{\cT}{\mathcal{T}}
\newcommand{\cH}{\mathcal{H}}
\newcommand{\cA}{\mathcal{A}}

\newcommand{\bB}{\mathbf{B}}

\newcommand{\hv}{\hat{v}}
\newcommand{\chv}{\hat{V}}
\newcommand{\hhv}{\tilde{v}}
\newcommand{\vc}[1]{{#1}}
\newcommand{\mat}[1]{\mathbf{#1}}

\newtheorem{theorem}{{Theorem}}

\title{\LARGE \bf
Linear Equilibria for Dynamic LQG Games with Asymmetric Information and Dependent Types
}

\author{Nasimeh Heydaribeni and Achilleas Anastasopoulos%
\thanks{This work was supported in part by NSF Grant ECCS-1608361.}
\thanks{The authors are with the Department of Electrical Engineering and Computer Science, University of Michigan, Ann Arbor, MI, 48105 USA {\tt\small {heydari,anastas}@umich.edu}}
}

\begin{document}

	\maketitle
\thispagestyle{empty}
\pagestyle{empty}
	

	\maketitle
	
	\begin{abstract}
We consider a non-zero-sum linear quadratic Gaussian (LQG) dynamic game with asymmetric information.
Each player observes privately a noisy version of a (hidden) state of the world $V$, resulting in dependent private observations. We study perfect Bayesian equilibria (PBE) for this game with equilibrium strategies that are linear in players' private estimates of $V$.
The main difficulty arises from the fact that players need to construct estimates on other players' estimate on $V$, which in turn would imply that an infinite hierarchy of estimates on estimates needs to be constructed, rendering the problem unsolvable.
We show that this is not the case: each player's estimate on other players' estimates on $V$ can be summarized into her own estimate on $V$ and some appropriately defined public information.
Based on this finding we characterize the PBE through a backward/forward algorithm akin to dynamic programming for the standard LQG control problem.
Unlike the standard LQG problem, however, Kalman filter covariance matrices, as well as some other required quantities, are observation-dependent and thus cannot be evaluated off-line through a forward recursion.
	\end{abstract}
	
	
	\section{Introduction}

Linear Quadratic Gaussian (LQG) models have been studied extensively for decision and control problems.
In the simplest instance of a single centralized controller it is well known that there is separation of estimation and control, posterior beliefs of the state are Gaussian, a sufficient statistic for control is the state estimate evaluated by the Kalman filter, the optimal control is linear in the state estimate, and the required covariance matrices can be calculated offline~\cite{KuVa86}.

The LQG model for the case of multiple controllers with different information patterns and single objective has also been studied extensively in the context of dynamic decentralized teams~\cite{ho1972team,yuksel2009stochastic, mahajan2015sufficient}. Although it is known that, in general, linear controllers are not optimal in LQG team problems~\cite{witsenhausen1968counterexample}, some information structures have been identified for which linear controllers are shown to be optimal~\cite{yuksel2009stochastic}.

In order to capture the strategic behavior of agents, which is an important aspect of today's  extensive networks \cite{HeAn18,heydaribeni2018distributed,HeAn18j}, LQG models have also been considered in the context of dynamic games. There is extensive literature on dynamic LQG games with asymmetric information, each work considering a different information structure, such as delayed observation sharing~\cite{basar1978two,basar1978decentralized}, or no access to other agents' observations~\cite{altman2009stochastic}, to name a few.  The appropriate solution for such problems is some notion of equilibrium such as Markov perfect equilibrium, Bayesian Nash equilibrium, perfect Bayesian equilibrium (PBE), sequential equilibrium, etc.~\cite{OsRu94, FuTi91,maskin2001markov,watson2017general}.
In dynamic games, due to the complexity of finding equilibrium strategies with increasing domains, researchers consider summaries of the agents' histories into time-invariant objects and define structured equilibria. 
For LQG models in particular, linear structured equilibria have been considered~\cite{gupta2014common,farokhi2014gaussian,VaAn16a,sayin2018dynamic}.

A broad classification of the relevant literature can be based on whether there is symmetric or asymmetric information among agents, whether a two-stage or a multi-stage game is considered, and whether the equilibrium concept used guarantees ``perfection'', i.e., sequential rationality at every possible (or impossible) information pattern.
Authors in~\cite{GuNaLaBa14b} have considered a multi-stage game with a special information structure enabling them to characterize a non-signaling Markov perfect equilibrium, which is a solution concept for symmetric information patterns. In~\cite{VaAn16a}, authors have considered a multi-stage game and characterized a signaling equilibrium which is linear in agents' private observations. In addition, a backward sequential decomposition was presented for the construction of the equilibrium, based on the general development in~\cite{VaSuAn15arxiv,vasal2018systematic}.
A number of works consider LQG games where information available to some players is affected by the decision of others.
The works of~\cite{crawford1982strategic} on strategic information transmission, and~\cite{farokhi2014gaussian} on Gaussian cheap talk consider two-stage games and focus on Bayesian Nash equilibria.
The classic work on Bayesian persuasion~\cite{kamenica2011bayesian}, and the related one on strategic deception~\cite{sayin2018dynamic} consider two-stage and multi-stage games, respectively, and focus on (sender preferred) subgame perfect equilibria owing to the fact that strategies of the sender are observed.

In this paper, we study a dynamic LQG non-zero-sum game with asymmetric information. We consider a model with an unknown Gaussian state of the world $V$, where each player $i$ has a private noisy observation $X^i_t$ of it at each time $t$. The private observations of players are conditionally independent given $V$. Our model closely follows that of~\cite{VaAn16a} with one important difference: the private observations of players in~\cite{VaAn16a} are independent where in our case, they are dependent through $V$; in particular they are conditionally independent given $V$. This model can also be thought of as a generalization of the one in~\cite{BiHiWe92} where $V$ models the value of a product (or a technology) and agents receive a noisy private signal about it and decide whether to adopt it or not, with the important difference that we allow multiple agents to act simultaneously and, unlike~\cite{BiHiWe92}, we also allow them to return to the marketplace at each time instance.
We hypothesize (and eventually prove) structured PBE with strategies for user $i$ being linear in $\hat{V}^i_t$, the private estimate of $V$ by user $i$, generated by a (private) Kalman filter.

What makes the considered model interesting and more challenging compared to previous works is that we need to deal with private beliefs while in most of the existing models, the beliefs are either public (e.g.,~\cite{gupta2014common,VaAn16a,vasal2018systematic}), or there is a public belief that can be easily augmented by the players' private signals to form the private beliefs (e.g.,~\cite{BiHiWe92,BiAn18,BiHeAn19j}), or even there are no beliefs to be defined due to the equilibrium concept used~\cite{crawford1982strategic,farokhi2014gaussian,kamenica2011bayesian,sayin2018dynamic}.

In order to intuitively explain the conceptual difficulty arising from the above dependent-signal model, consider the following thought process. If a player acts according to her private estimate $\hat{V}^i_t$ of the hidden variable $V$ and she expects other players to behave in the same way, she needs to form a belief over other players' beliefs to interpret and predict their actions and she has to take that belief into account when acting. In other words, she has to form a belief over (at least) $\hat{V}^j_t$ for all other users $j\neq i$. This is also a form of a belief on beliefs which is also a private information of user $i$ and it has to be taken into account in her strategies, and one may expect that in the simplest case this will reduce to user $i$'s estimates $\tilde{V}^{i,j}_t$ of $\hat{V}^j_t$. Due to symmetry of the information structure, all other players should do the same. But now, it is clear that user $i$ needs to form beliefs over beliefs over beliefs of other players.
In the simplest case this would involve user $i$'s private estimates $\tilde{V}^{i,j,k}_t$ of the estimates $\tilde{V}^{j,k}_t$.
This chain continues as long as this hierarchy of beliefs are private. It stops whenever the beliefs in one step are public or public functions of previous step beliefs.
One of the main contributions of this paper is to show that, due to the conditional independence of the private signals given $V$, this chain stops at the second step and players estimations over the estimations of others, $\tilde{V}^{i,j}_t$, are public linear functions of their own estimations (the first step beliefs), $\hat{V}^i_t$.

Once the above task is accomplished, we show that the equilibrium strategies can be characterized by an appropriate backward sequential decomposition algorithm akin to dynamic programming. The main difference from the standard stochastic control LQG framework is that the forward recursion that evaluates covariance matrices cannot be performed separately as it depends on the equilibrium strategies. This was also the case in~\cite{VaAn16a}.
A unique feature of this work is the requirement to update in a forward manner additional quantities that are observation dependent (public actions). This precludes off-line evaluation of these forward-updated quantities  and necessitates their inclusion as part of the state of the above mentioned backward sequential decomposition.
This is the second main contribution of this work.

The remaining part of the paper is structured as follows. In Section \ref{section2}, the model is described. Section \ref{section3} is a review over the solution concept that we have considered in this paper. We have developed our main results in Section \ref{section4}. Section \ref{section5} summarizes the arguments in Section \ref{section4} into two algorithms and we conclude in Section \ref{section6}.

\subsection{Notation}
		We use upper case letters for scalar and vector random variables and lower case letters for their realizations. Bold upper case letters are used to denote matrices. Subscripts denote time indices and superscripts represent player identities. The notation $-i$ denotes the set of all players except $i$. All vectors are column vectors. The transpose of a matrix $\mat{A}$ (or vector) is denoted by $\mat{A}'$.  We use  semicolons $``;"$ for vertical concatenation of matrices (or vectors). For any vector (or matrix) with time and player indices, $a^i_t$ (or $\mat{A}^i_t$), $a^{-i}_t$ denotes the vertical concatenation of vectors (or matrices) $a^1_t,a^2_t,\hdots,a^{i-1}_t,a^{i+1}_t,\hdots$. Further, $a^i_{1:t}$ means  $(a^i_1,a^i_2,\hdots,a^i_t)$. In general, for any vector with time and player indices, $a^i_t$, we remove the superscript to show the vertical concatenation of the whole vectors and we remove the subscript to show the set of all vectors for all times. The notations $\mat{0}$ is  used to show the matrix of all zeros with appropriate dimension and $\mat{I}$ denotes the identity matrix of appropriate dimension. For two matrices $\mat{A}$ and $\mat{B}$, $\mathfrak{D}(\mat{A},\mat{B})$ represents the block diagonal concatenation of these matrices, i.e., $\left[\begin{array}{cc}
		\mat{A} & \mat{0} \\ \mat{0} & \mat{B}
		\end{array}\right]$ (it applies for any number of matrices). By $\mathfrak{D}(\mat{A}^{-i})$, we mean the block diagonal concatenation of matrices $\mat{A}^j$ for $j \in -i$. Further, $quad(A;B)$ represents $B'AB$.
		For the equation $\left[\begin{array}{ccc}\tilde{a} \ ;& \tilde{b} \ ;&  \tilde{c}
		\end{array}\right]=\mat{A}\left[\begin{array}{ccc}a  \ ;&  b \ ;&  c
		\end{array}\right]$, the notation $(\mat{A})_{\tilde{a},b}$ denotes the submatrix of $\mat{A}$ corresponding to rows $\tilde{a}$ and columns $b$. We use $``:"$ for either of the row or column subscripts to indicate the whole rows or columns.
		We use $\delta(\cdot)$ for the Dirac delta function. For any Euclidean set $\mathcal{S}$, $\Delta(\mathcal{S})$ represents the space of all probability measures on $\mathcal{S}$.

\section{Model} \label{section2}
	
	We consider a discrete time dynamic system with $\cN=\{1,2,...,N\}$ strategic players over a finite time horizon $\cT=\{1,2,...,T\}$.   There is a static unknown state of the world $V\sim N(\mat{0},\mat{\Sigma})$ with size $N_v$. Each player has a private noisy observation $X^i_t$ of $V$ at every time step $t\in\cT$
	\vspace{0.1cm}
	\begin{equation}
	x^i_t=v+w^i_t,
	\end{equation}
	where $W^i_t \sim N(\mat{0},\mat{Q}^i)$ and all of the noise random vectors $W^i_t$ are independent across $i$ and $t$ and also of $V$. The values of $\mat{\Sigma}$ and $\mat{Q}^i, \ \forall i \in \cN$ are common knowledge between players. Also, we assume that players have perfect recall. At time $t$, player $i$ takes action $a^i_t \in \cA=\mathbb{R}^{N_a}$ which is observed publicly by all players. We can construct the history of the system as $h_t=(v,x_{1:t},a_{1:t-1}) \in \cH_t $ and the history observed by player $i$ as $h^i_t=(x^i_{1:t},a_{1:t-1}) \in \cH^i_t$. At the end of time step $t$, each player $i$ receives the reward $R^i(v,a_t)$,
	\vspace{0.1cm}
	\begin{equation}
	R^i(v,a_t)=\left[{\begin{array}{cc}\hspace{-0.1cm} v' & a_t' \hspace{-0.1cm} \end{array}}\right]
\bB^i
\left[\begin{array}{c} v \\ a_t \end{array}\right]\hspace{-0.1cm} =quad(\bB^i;\left[\begin{array}{c} v \\ a_t \end{array}\right] ),
	\end{equation}
	where $B^i$ is a symmetric matrix of appropriate dimensions. We assume that the rewards are not observed by the players until the end of the time horizon.
	
 Let $g^i=(g^i_t)_{t \in \cT}$ be a probabilistic strategy of player $i$, where $g^i_t: \cH^i_t \rightarrow \Delta(\cA)$ such that player $i$'s action is generated according to the distribution $g^i_t(\cdot|h^i_t)$. The strategy profile of all players is denoted by $g$. For the strategy profile $g$, player $i$'s total expected reward is
\vspace{0.1cm}
\begin{equation}
J^{i,g}:=\mathbb{E}^g\left\{\sum_{t=1}^TR^i(V,A_t)\right\},
\end{equation}	
and her objective is to maximize her total expected reward.

\section{Solution concept}\label{section3}

	We can model this system as a dynamic game with asymmetric information and an appropriate solution concept for such games is the perfect Bayesian equilibrium (PBE). A PBE consists of a pair $(\beta,\mu)$ (an assessment) of strategy profile $\beta=(\beta^i_t)_{t\in \cT , i \in \cN}$ and belief system  $\mu=(\mu^i_t)_{t\in \cT , i \in \cN}$ where $\mu^i_t : \cH^i_t \rightarrow \Delta(\cH_t)$ satisfies Bayesian updating\footnote{Bayesian updating includes both on and off equilibrium histories.  This condition requires the beliefs to be Bayesian updated given any history, whether that history is on equilibrium or off equilibrium~\cite{watson2017general}.} and sequential rationality holds. For any $i \in \cN, \  t \in \cT, h^i_t \in \cH^i_t, \tilde{\beta}^i$, sequential rationality imposes the following condition:
\begin{equation}
	\begin{aligned}
	\mathbb{E}^{\beta^i\beta^{-i}}_{\mu}\left\{\sum_{n=t}^T R^i(V,A_n)|h^i_t\right\} \geq   \mathbb{E}^{\tilde{\beta}^i\beta^{-i}}_{\mu}\left\{\sum_{n=t}^T R^i(V,A_n)|h^i_t\right\}
	\end{aligned}
	\end{equation}

\section{Structured PBE}\label{section4}
The strategies $g^i_t(\cdot|h^i_t)$ have a domain that is expanding in time. Finding such strategies is complicated with the complexity growing exponentially with the time horizon. For this reason, we consider summaries for $h^i_t \in \cH^i_t$, i.e., $S(h^i_t)$, that are not expanding in time. We are interested in PBEs with strategies, $g^i_t(\cdot|h^i_t)=\psi^i_t(\cdot|S(h^i_t))$, that are functions of $h^i_t$ only through the summaries $S(h^i_t)$. These PBEs are called structured PBEs~\cite{vasal2018systematic}. In contrast to $\cH^i_t$,  the set of summaries does not grow in time and therefore, finding such structured PBEs is less complicated than a general PBE.  According to \cite{vasal2018systematic}, we  can show that players can guarantee the same rewards by playing structured strategies compared to the general non-structured ones. In the dynamic games with asymmetric information, the summaries are usually the belief of players over the unknown variables of the game. In this model, we will show that these beliefs are Gaussian and since any Gaussian belief can be expressed in terms of its mean and covariance matrix, we define the summaries such that they include the mean and covariance matrices of the beliefs of the players over $V$. The mean of each player's belief, i.e., her estimation over $V$, will be her private information. The covariance matrix, however, can be calculated publicly.
Each player, in addition to her own estimate of $V$, needs to interpret actions of others and predict their future actions. Hence, each player needs to have a belief over the estimates of other players on $V$. We will show that this latter belief is also Gaussian and therefore, one needs to keep track of only its mean and covariance. Therefore, for all $ i \in \cN, \ t \in \cT$, we define the following quantities,
\vspace{0.2cm}
\begin{align}
\hv^i_t&=\mathbb{E}[V| h^i_t]=\mathbb{E}[V|x^i_{1:t},a_{1:t-1}] \\
\hhv^{i,j}_t&=\mathbb{E}[\chv^j_t|h^i_t]=\mathbb{E}[\chv^j_t|x^i_{1:t},a_{1:t-1}].
\end{align}	
The quantity $\hv^i_t$ is player $i$'s best estimate of $V$ given her observations up to time $t$. As mentioned before, this quantity is a private estimation for player $i$ and is not measurable with respect to the sigma algebra generated by the observations of any other player $j$. Hence, player $i$ should form an estimation over the private estimations of other players and this is why $\hhv^{i,j}_t$ is defined. This in turn implies that players' strategies should also be a function of their estimations over others' estimations. Hence, the same argument holds about the need to define an estimation over estimations of players over other players' estimations. This argument continues as long as these estimations are private. This chain stops whenever one of the estimations of players is public (or a public function of previous-step private estimations) and therefore, there is no need to form an estimation over it.

Indeed, we will show that $\hhv^{i,-i}_t$ is a linear function of $\hv^i_t$, hence,  there is no need to include $\hhv^{i,-i}_t$ in the summary $S(h^i_t)$ and therefore, no other player needs to form an estimation over it. The summary we use for $h^i_t$ is defined as $S(h^i_t)=(\hv^i_t, P(h^i_t))$, where $P(h^i_t)$ is the public summary for $h^i_t$ and it includes the covariance matrix of player $i$'s belief over $V$ and some other needed quantities that will be subsequently defined.
We are interested in the strategies $A^i_t \sim \psi^i_t(\cdot|\hv^i_t,P(h^i_t))=\gamma^i_t(\cdot|\hv^i_t)$, where $\gamma^i_t=\theta^i_t(P(h^i_t))$. In particular, we want to prove that the linear strategies $\gamma^i_t(a^i_t|\hv^i_t)=\delta(a^i_t-\mat{L}^i_t\hv^i_t-\vc{m}^i_t)$, where $\mat{L}^i_t$ and $\vc{m}^i_t$ are a matrix and a vector with appropriate dimensions and are functions of $P(h^i_t)$, form a PBE of the game.

\subsection{State Evolution}
In order to prove Gaussianity of the beliefs over $V$ and other players' estimations, for each player $i$, we define a state vector that includes $v$ and all of the players' estimations in addition to her private observation. We will use Kalman filter results to update this vector recursively in time and prove Gaussianity and other properties for it. We define the state vector as $s^i_t=\left[\begin{array}{cccc} v \ ; & \hv^i_{t-1} \ ; & \hv^{-i}_{t-1} \ ; & x^i_t \end{array}\right]$, for each player $i \in \cN$.
By deriving the conditional distribution of the state vector $S^i_t$ given the observation of player $i$, we can form her belief over $V$ and other players' estimations $\chv^{-i}_{t-1}$.

In the next theorem, we show that for  $\gamma^j_k(a^j_k|\hv^j_k)=\delta(a^j_k-\mat{L}^j_k\hv^j_k-\vc{m}^j_k), \ \forall k \leq t-1, j \in \cN$,  $\hhv^{i,j}_t$ is a linear function of $\hv^i_t$. Further, the state $s^i_t$ is updated recursively in terms of $s^i_{t-1}$ through a Gauss-Markov model.

\begin{theorem}\label{th1}
For $\gamma^j_k(a^j_k|\hv^j_t)=\delta(a^j_k-\mat{L}^j_k\hv^j_k-\vc{m}^j_k), \ \forall k \leq t-1, j \in \cN$,

\begin{enumerate}[noitemsep,wide=0pt, leftmargin=0pt]
\item[(a)] The random vector $s^i_t=\left[\begin{array}{cccc} v \ ; & \hv^i_{t-1} \ ; & \hv^{-i}_{t-1} \ ; & x^i_t \end{array}\right]$ evolves according to a linear Gaussian process,
\begin{equation}
\left[\begin{array}{c} v\\ \hv^i_t\\ \hv^{-i}_t \\x^i_{t+1} \end{array}\right] = \mat{A}^i_t \left[\begin{array}{c} v\\ \hv^i_{t-1}\\ \hv^{-i}_{t-1} \\x^i_t \end{array}\right]+\mat{H}^i_t \left[\begin{array}{c}
w^{-i}_t \\ w^i_{t+1}
\end{array}\right] +\vc{d}^i_t,
\label{update}
\end{equation}
where
\begin{equation}
\mat{A}^i_t=\left[\begin{array}{c} \begin{array}{cccc}
\mat{I} & \mat{0} & \mat{0} & \mat{0}
\end{array}\\  \mat{G}^{i,i}_t \\   \mat{G}^{i,-i}_t\\ \begin{array}{cccc}
\mat{I} & \mat{0} & \mat{0} & \mat{0}
\end{array}\end{array}\right]
\end{equation}
and  $\mat{G}^{i,i}_t$, $\mat{G}^{i,-i}_t$, $\mat{H}^i_t$ and $\vc{d}^i_t$ are matrices and vector with appropriate dimensions (they will be constructed in the proof).
\item[(b)] The conditional expectation $\mathbb{E}[\chv^{-i}_t|v,a^{-i}_{1:t-1}]$ is a linear function of $v$,
\begin{equation}
\mathbb{E}[\chv^{-i}_t|v,a^{-i}_{1:t-1}]=\mat{E}^i_t v+\vc{f}^i_t,
\label{vhhupdatecond}
\end{equation}
and  $\mat{E}^{i}_t$, and $\vc{f}^i_t$ are a matrix and a vector, respectively, with appropriate dimensions (they will be constructed in the proof).

\end{enumerate}

\end{theorem}

Before proving this Theorem we note that part $(b)$ of Theorem \ref{th1} implies that the estimation of player $i$ over private estimations of players $-i$, i.e, $\hhv^{i,-i}_t$, is a linear function of $\hv^i_t$,
\begin{equation}
\hhv^{i,-i}_t=\mat{E}^i_t \hv^i_t+\vc{f}^i_t.
\label{vhhupdate}
\end{equation}

\begin{proof}
Equation \eqref{update} is obvious for the first and fourth part of the state ($v$ and $x^i_{t+1}$) by setting $(\mat{H}^i_t)_{x^i_{t+1},:}=\left[\begin{array}{cc} \mat{0} & \mat{I}
\end{array}\right]$, $(\mat{H}^i_t)_{v,:}=\mat{0}$ and $(\vc{d}^i_t)_{vx^i_{t+1}}= \mat{0}$.
We prove all other parts of Theorem \ref{th1} together through induction.
\begin{itemize}[noitemsep,wide=0pt, leftmargin=0pt]
\item Induction basis: for $t=1$, we have $s^i_1=\left[\begin{array}{cccc} v \ ; & \hv^i_0 \ ; & \hv^{-i}_0\ ; & x^i_1 \end{array}\right]=\left[\begin{array}{cccc} v \ ; & \mat{0} \ ;& \mat{0} \ ;  & x^i_1 \end{array}\right]$ and for $t=2$, $s^i_2=\left[\begin{array}{cccc} v \ ; & \hv^i_1 \ ; & \hv^{-i}_1 \ ; & x^i_2 \end{array}\right]$. The definition of $\hv^i_1$ and the fact that the vectors $v$ and $x^i_1$ are jointly Gaussian results in the following \cite[Ch.7]{KuVa86},
\begin{equation}
\begin{aligned}
\hv^i_1&=\mathbb{E}[V|x^i_1]=\mathbb{E}[V]+\mathbb{E}[VX^{i'}_1]{\mathbb{E}[X^i_1X^{i'}_1]}^{-1}(x^i_1-\mathbb{E}[X^i_1])\\ &=\mat{\Sigma}{(\mat{\Sigma}+\mat{Q}^i)}^{-1}x^i_1={\mat{\Sigma}}{(\mat{\Sigma}+\mat{Q}^i)}^{-1}(v+w^i_1).
\label{basisV}
\end{aligned}
\end{equation}
We can also write 
\begin{equation}
\hv^j_1=\mat{\Sigma}{(\mat{\Sigma}+\mat{Q}^j)}^{-1}(v+w^j_1), \quad \forall j \in \cN.
\end{equation}
Therefore, we can derive $\mat{A}^i_1$ (and essentially matrices $\mat{G}^{i,i}_1, \ \mat{G}^{i,-i}_1 $), $\mat{H}^i_1$ and $\vc{d}^i_1$
\begin{align}
&\mat{A}^i_1=\left[\begin{array}{cccc} \mat{I} & \mat{0} & \mat{0} & \mat{0}\\ \mat{0} & \mat{0} & \mat{0} & \mat{\Sigma}{(\mat{\Sigma}+\mat{Q}^i)}^{-1} \\\mat{\Sigma}{(\mat{\Sigma}+\mat{Q}^{-i})}^{-1} & \mat{0} & \mat{0} & \mat{0}  \\ \mat{I} & \mat{0}& \mat{0}& \mat{0}\end{array}\right] \\
&\Rightarrow \mat{G}^{i,i}_1=\left[\begin{array}{cccc}  \mat{0} & \mat{0} & \mat{0} & \mat{\Sigma}{(\mat{\Sigma}+\mat{Q}^i)}^{-1} \end{array}\right]\\
& \mat{G}^{i,-i}_1=\left[\begin{array}{cccc}  \mat{\Sigma}{(\mat{\Sigma}+\mat{Q}^{-i})}^{-1} & \mat{0} & \mat{0} & \mat{0}  \end{array}\right]\\
&\mat{H}^i_1=\left[\begin{array}{cc}  \mat{0} & \mat{0} \\ \mat{0} & \mat{0}\\ \mathfrak{D}(\mat{\Sigma}{(\mat{\Sigma}+\mat{Q}^{-i})}^{-1}) & \mat{0}  \\ \mat{0} & \mat{I}\end{array}\right]\\
&\vc{d}^i_1=\mat{0},
\end{align}
where $\mat{\Sigma}{(\mat{\Sigma}+\mat{Q}^{-i})}^{-1}$ is the vertical concatenation of the matrices $\mat{\Sigma}{(\mat{\Sigma}+\mat{Q}^j)}^{-1}$ for $j \in -i$.
Further, we can derive the estimation of player $i$ over other players' estimations as follows,
\begin{equation}
\begin{aligned}
\hhv^{i,j}_1&=\mathbb{E}[\chv^j_1|x^i_1]=\mathbb{E}[\mat{\Sigma}{(\mat{\Sigma}+\mat{Q}^j)}^{-1}(V+W^j_1)|x^i_1]\\&=\mat{\Sigma}{(\mat{\Sigma}+\mat{Q}^j)}^{-1}\mathbb{E}_V[\mathbb{E}[V+W^j_1|x^i_1,V]|x^i_1]\\&=\mat{\Sigma}{(\mat{\Sigma}+\mat{Q}^j)}^{-1}\mathbb{E}[V|x^i_1]=\mat{\Sigma}{(\mat{\Sigma}+\mat{Q}^j)}^{-1}\hv^i_1,
\end{aligned}
\end{equation}
which means that
\begin{equation}
\begin{aligned}
\mat{E}^i_1&=\mat{\Sigma}{(\mat{\Sigma}+\mat{Q}^{-i})}^{-1}\\
\vc{f}^i_1&=\mat{0}.
\label{basisEF}
\end{aligned}
\end{equation}
This concludes the proof of part $(a)$ and $(b)$ of the theorem for $t=1$.

\item Induction hypothesis: \eqref{update} and \eqref{vhhupdatecond} hold for $t=k-1$ and $k \geq 2$.

\item Induction step:
we first show one important result from the induction hypothesis for part (b) of the theorem. Notice that due to conditional independence of $x^j_{k-1}$'s given $v$ across time and players, and since $\hv^j_{k-1}$ is a function of $x^j_{1:k-1}$ and  $a_{1:k-2}$, and since $a^i_{1:k-2}$ is a function of $x^i_{1:k-2}$ and $a^{-i}_{1:k-3}$, we have
\begin{equation}
\begin{aligned}
\hhv^{i,j}_{k-1}&=\mathbb{E}[\chv^j_{k-1}|x^i_{1:k-1},a_{1:k-2}]\\
&=\mathbb{E}_V[\mathbb{E}[\chv^j_{k-1}|V,x^i_{1:k-1},a_{1:k-2}]|x^i_{1:k-1},a_{1:k-2}]\\
&=\mathbb{E}_V[\mathbb{E}[\chv^j_{k-1}|V,a^{-i}_{1:k-2}]|x^i_{1:k-1},a_{1:k-2}]\\
&= \mathbb{E}_V[\mat{E}^i_{k-1} V+\vc{f}^i_{k-1}|x^i_{1:k-1},a_{1:k-2}]\\
&=\mat{E}^i_{k-1} \mathbb{E}[V|x^i_{1:k-1},a_{1:k-2}]+\vc{f}^i_{k-1}\\
&=\mat{E}^i_{k-1}\hv^i_{k-1}+\vc{f}^i_{k-1}.
\label{cond}
\end{aligned}
\end{equation}

In order to prove the results for $t=k$, by using the induction hypothesis, we form a linear Gaussian model with partial observations and use Kalman filter results~\cite[Ch.7]{KuVa86}. Consider the following stochastic system with
state $s^i_k$, state evolution given by \eqref{update} (for $t=k-1$)
\begin{subequations}
	\vspace{0.1cm}
\begin{equation}
s^i_k = \mat{A}^i_{k-1} s^i_{k-1} +\mat{H}^i_{k-1} \left[\begin{array}{cc}
w^{-i}_{k-1}\\  w^i_k
\end{array} \right]+\vc{d}^i_{k-1},\\
\end{equation}
and observation given by
\vspace{0.1cm}
\begin{equation}
y^i_k=\left[\begin{array}{c} a^i_{k-1}-m^i_{k-1}\\  a^{-i}_{k-1}-m^{-i}_{k-1}\\ x^i_k \end{array}\right]=\mat{C}^i_ks^i_k,
\end{equation}
where
\begin{equation}
\mat{C}^i_k=\left[\begin{array}{cccc}\mat{0} & \mat{L}^i_{k-1} & \mat{0} & \mat{0} \\ \mat{0} &   \mat{0} & \mathfrak{D}(\mat{L}^{-i}_{k-1})   & \mat{0}\\ \mat{0} &
\mat{0}& \mat{0}
 & \mat{I}\end{array}\right].
\end{equation}
\label{LQG}
\end{subequations}

Note that $y^i_{1:k}$ is a shifted version of $h^i_k$.
We denote $\mathbb{E}[{S^i_k|y^i_{1:k}}]$ and $\mathbb{E}[{S^i_k|y^i_{1:k-1}}]$ by $s^i_{k|k}$ and $s^i_{k|k-1}$, respectively. By using standard Kalman filter results~\cite[Ch.7]{KuVa86},
we have
\begin{equation}
\begin{aligned}
&s^i_{k|k}=\mathbb{E}[{S^i_k|x^i_{1:k},a_{1:k-1}}]=\left[\begin{array}{c} \hv^i_k \\ \hv^i_{k-1}\\  \mathbb{E}[{\chv^{-i}_{k-1}|x^i_{1:k},a_{1:k-1}}] \\ x^i_k \end{array}\right]\\
&= \mat{A}^i_{k-1}s^i_{k-1|k-1}+\mat{J}^i_k( y^i_k-\mat{C}^i_ks^i_{k|k-1})+\vc{d}^i_{k-1}\\
& \Rightarrow \hv^i_k=\hv^i_{k-1}+(\mat{J}^i_k)_{\hv^i_k,:}(y^i_k-\mat{C}^i_ks^i_{k|k-1})\\&= \hv^i_{k-1}+(\mat{J}^i_k)_{\hv^i_k,:}\left[\begin{array}{c} a^i_{k-1}-m^i_{k-1}-\mat{L}^i_{k-1}\hv^i_{k-1} \\ a^{-i}_{k-1}-m^{-i}_{k-1}-\mathfrak{D}(\mat{L}^{-i}_{k-1})\hhv^{i,-i}_{k-1} \\x^i_k -\mathbb{E}[X^i_k|x^i_{1:k-1},a_{1:k-2}] \end{array}\hspace{-0.1cm}\right]\hspace{-0.1cm},
\end{aligned}
\end{equation}
where
\begin{equation}
\mat{J}^i_k=\mat{\Sigma}^i_{k|k-1}\mat{C}^{i'}_k(\mat{C}^i_k \mat{\Sigma}^i_{k|k-1} \mat{C}^{i'}_k )^{-1},
\label{J}
\end{equation}
\begin{equation}
\mat{\Sigma}^i_{k|k-1}=\mat{A}^i_{k-1}\mat{\Sigma}^i_{k-1}\mat{A}^{i'}_{k-1}+\mat{H}^i_{k-1} \mathfrak{D}(\mat{Q}^{-i},\mat{Q}^i) \mat{H}^{i'}_{k-1},
\label{condsigma}
\end{equation}
and $\mat{\Sigma}^i_k$ is the covariance matrix of $S^i_k$ conditioned on $h^i_k$ and according to \cite[Ch.7]{KuVa86}, it is derived from the following recursive update equation
\vspace{0.1cm}
\begin{equation}
	\begin{aligned}
	&\mat{\Sigma}^i_{k+1}\hspace{-0.1cm}=\hspace{-0.05cm}(\mat{I}-\mat{J}^i_{k+1}\mat{C}^i_{k+1})(\mat{A}^i_k\mat{\Sigma}^i_k\mat{A}^{i'}_k\hspace{-0.1cm}+\hspace{-0.05cm}\mat{H}^i_k \mathfrak{D}(\mat{Q}^{-i}\hspace{-0.1cm},\mat{Q}^i) \mat{H}^{i'}_k)\\
	&\mat{\Sigma}^i_1=\mathbb{E}[S^i_1S^{i'}_1]-\mathbb{E}[S^i_1 X^{i'}_1](\mathbb{E}[X^i_1 X^{i'}_1])^{-1}\mathbb{E}[S^i_1 X^{i'}_1]'\\
&=\left[
\begin{array}{cccc} \mat{\Sigma} & \mat{0} & \mat{0} & \mat{\Sigma}\\
\mat{0} & \mat{0} & \mat{0} & \mat{0}\\
\mat{0} & \mat{0} & \mat{0} & \mat{0}\\
\mat{\Sigma} & \mat{0} & \mat{0} & \mat{\Sigma}+\mat{Q}^i
\end{array}\right]\\ &-
\left[\hspace{-0.1cm}\begin{array}{c} \mat{\Sigma} \\
\mat{0} \\
\mat{0}\\
\mat{\Sigma}+\mat{Q}^i
\end{array}\hspace{-0.1cm}\right]
\hspace{-0.05cm}(\mat{\Sigma}+\mat{Q}^i)^{-1}\hspace{-0.05cm}\left[\hspace{-0.1cm}\begin{array}{cccc} \mat{\Sigma}' & \mat{0} & \mat{0} &	(\mat{\Sigma}+\mat{Q}^i)'
\end{array}\hspace{-0.1cm}\right].
	\end{aligned}
	\label{Sigma}
	\end{equation}
	Notice that unlike $\hv^i_t$, which is private information of player $i$, the matrix $\mat{\Sigma}^i_t$ is a public quantity due to the independence of equation \eqref{Sigma} to the private observations of player $i$. 

Since we can write 
\begin{equation}
\begin{aligned}
\mathbb{E}[X^i_k|x^i_{1:k-1},a_{1:k-2}]&=\mathbb{E}[V+W^i_k|x^i_{1:k-1},a_{1:k-2}]\\
&=\mathbb{E}[V|x^i_{1:k-1},a_{1:k-2}]=\hv^i_{k-1},
\end{aligned}
\end{equation}
 and according to \eqref{cond},
\begin{equation}
\begin{aligned}
&\hv^i_k=\hv^i_{k-1}+(\mat{J}^i_k)_{\hv^i_k,:}\left[\begin{array}{c} \mat{0} \\ -\mathfrak{D}(\mat{L}^{-i}_{k-1})\mat{E}^i_{k-1} \hv^i_{k-1} \\ x^i_k -\hv^i_{k-1} \end{array}\right]\\
&+ (\mat{J}^i_k)_{\hv^i_k,a^{-i}_{k-1}} (a^{-i}_{k-1}-m^{-i}_{k-1}-\mathfrak{D}(\mat{L}^{-i}_{k-1})\vc{f}^i_{k-1})  \\&
= \mat{G}^{i,i}_k \left[\begin{array}{c} v\\ \hv^i_{k-1}\\ \hv^{-i}_{k-1} \\x^i_k \end{array}\right]+(\vc{d}^i_k)_{\hv^i_k},
\label{vhupdate}
\end{aligned}
\end{equation}
where
\begin{equation}
\begin{aligned}
&(\mat{G}^{i,i}_k)_{:,v\hv^{-i}_{k-1}}=\mat{0}\\ &(\mat{G}^{i,i}_k)_{:,x^i_k}=(\mat{J}^i_k)_{\hv^i_k,x^i_k}\\
&(\mat{G}^{i,i}_k)_{:,\hv^i_{k-1}}\hspace{-0.1cm} =\hspace{-0.05cm} \mat{I}-\hspace{-0.05cm}(\mat{J}^i_k)_{\hv^i_k,a^{-i}_{k-1}}\mathfrak{D}(\mat{L}^{-i}_{k-1})\mat{E}^i_{k-1}\hspace{-0.05cm}-\hspace{-0.05cm}(\mat{J}^i_k)_{\hv^i_k,x^i_k} \\
& (\vc{d}^i_k)_{\hv^i_k}= (\mat{J}^i_k)_{\hv^i_k,a^{-i}_{k-1}} (a^{-i}_{k-1}-m^{-i}_{k-1}-\mathfrak{D}(\mat{L}^{-i}_{k-1})\vc{f}^i_{k-1})\\
&(\mat{H}^i_k)_{\hv^i_k,:}=\mat{0}.
\end{aligned}
\end{equation}
By considering  the dynamic system \eqref{LQG} for each of the  players $-i$, we can write \eqref{vhupdate} for players $-i$. Since $x^{-i}_k$ is not part of $s^i_k$, we can substitute it by $v+w^{-i}_k$ and derive $\mat{G}^{i,j}_k$ and $(\mat{H}^i_k)_{\hv^j_k,:}$ for all $j \in -i$ as follows,
\begin{equation}
\begin{aligned}
&(\mat{G}^{i,j}_k)_{:,v}=(\mat{J}^j_k)_{\hv^j_k,x^j_k} \\ &(\mat{G}^{i,j}_k)_{:,\hv^{-j}_{k-1}x^i_k}=\mat{0}\\
&(\mat{G}^{i,j}_k)_{:,\hv^j_{k-1}} \hspace{-0.1cm}=\hspace{-0.05cm} \mat{I}-\hspace{-0.05cm}(\mat{J}^j_k)_{\hv^j_k,a^{-j}_{k-1}}\mathfrak{D}(\mat{L}^{-j}_{k-1})\mat{E}^j_{k-1}\hspace{-0.05cm}-\hspace{-0.05cm}(\mat{J}^j_k)_{\hv^j_k,x^j_k} \\
& (\vc{d}^i_k)_{\hv^j_k}= (\mat{J}^j_k)_{\hv^j_k,a^{-j}_{k-1}} (a^{-j}_{k-1}-m^{-j}_{k-1}-\mathfrak{D}(\mat{L}^{-j}_{k-1})\vc{f}^j_{k-1})\\
&(\mat{H}^i_k)_{\hv^{-i}_k,:}=\left[\begin{array}{cc}
\mathfrak{D}((\mat{J}^{-i}_k)_{\hv^{-i}_k,x^{-i}_k})& \mat{0}
\end{array}\right].
\end{aligned}
\end{equation}

Therefore, we have derived the matrices $\mat{A}^i_k$, $\mat{H}^i_k$ and $\vc{d}^i_k$ and so \eqref{update} holds for $t=k$.

Next, we prove \eqref{vhhupdatecond} for $t=k$. We use the fact that  observations of players are independent conditioned on $V$ and consider a conditional linear Gaussian model.
Note that the inner expectation in \eqref{cond} is publicly measurable conditioned on $V$. We use this fact to form a conditional model, where the observations are the conditions in the inner expectation in \eqref{cond}, and we derive conditional Kalman filters. Consider the following linear Gaussian model for $t=k$, with
\begin{equation}
\hspace{-3cm}
\begin{aligned}
&\text{state}\\
&\tilde{s}^i_{k}=\left[\begin{array}{c} v\\ \hv^{-i}_{k-1} \end{array}\right], \\
&\text{state evolution}\\
&\tilde{s}^i_{k+1} = \tilde{\mat{A}}^i_k \tilde{s}^i_k +\tilde{\mat{H}}^i_k w^{-i}_k+\tilde{\vc{d}}^i_k,\\
&\text{and observation}\\
&\tilde{y}^i_k=\left[\begin{array}{c}v \\ a^{-i}_{k-1}-m^{-i}_{k-1} \end{array}\right]=\tilde{\mat{C}}^i_ks^i_k,\\
\end{aligned}
\end{equation}
where
\begin{align}
&\tilde{\mat{A}}^i_k=\left[\begin{array}{cc} \mat{I}  & \mat{0} \\ & \hspace{-0.5cm}  \tilde{\mat{G}}^{i,-i}_k \end{array}\right]\\
&\tilde{\mat{G}}^{i,-i}_k =(\mat{G}^{i,-i}_k)_{:,v\hv^{-i}_{k-1}}\\
&\tilde{\mat{C}}^i_k=\left[\begin{array}{cc}\mat{I} &  \mat{0} \\ \mat{0} & \quad \mathfrak{D}(\mat{L}^{-i}_{k-1})  \end{array}\right]\\
&\tilde{\mat{H}}^i_k= (\mat{H}^i_k)_{v\hv^{-i}_k,w^{-i}_k}\\
&\tilde{\vc{d}}^i_k=(\vc{d}^i_k)_{v\hv^{-i}_k}.
\end{align}
By using Kalman filter results and the induction hypothesis we can write
\begin{equation}
\begin{aligned}
&\tilde{s}^i_{k+1|k}=\mathbb{E}[\tilde{S}^i_{k+1}|\tilde{y}^i_{1:k}]=\mathbb{E}[\tilde{S}^i_{k+1}|v,a^{-i}_{1:k-1}]\\
&=\tilde{\mat{A}}^i_k\tilde{s}^i_{k|k-1}+\tilde{\mat{A}}^i_k\tilde{\mat{J}}^i_k(\tilde{y}^i_k-\tilde{\mat{C}}^i_k\tilde{s}^i_{k|k-1})+\tilde{\vc{d}}^i_k\\
& \Rightarrow  \mathbb{E}[\chv^{-i}_k|v,a^{-i}_{1:k-1}]=(\mat{G}^{i,-i}_k)_{:,v}v\\&+(\mat{G}^{i,-i}_k)_{:,\hv^{-i}_{k-1}} \mathbb{E}[\chv^{-i}_{k-1}|v,a^{-i}_{1:k-2}]\\&-(\tilde{\mat{A}}^i_k\tilde{\mat{J}}^i_k)_{\hv^{-i}_k,a^{-i}_{k-1}}\mathfrak{D}(\mat{L}^{-i}_{k-1})\mathbb{E}[\chv^{-i}_{k-1}|v,a^{-i}_{1:k-2}]\\&
+(\tilde{\mat{A}}^i_k\tilde{\mat{J}}^i_k)_{\hv^{-i}_k,a^{-i}_{k-1}}(a^{-i}_{k-1}-m^{-i}_{k-1})+(\vc{d}^i_k)_{\hv^{-i}_k}\\
&=(\mat{G}^{i,-i}_k)_{:,v}v+(\mat{G}^{i,-i}_k)_{:,\hv^{-i}_{k-1}}(\mat{E}^i_{k-1}v+\vc{f}^i_{k-1})\\&-(\tilde{\mat{A}}^i_k\tilde{\mat{J}}^i_k)_{\hv^{-i}_k,a^{-i}_{k-1}} \mathfrak{D}(\mat{L}^{-i}_{k-1})(\mat{E}^i_{k-1}v+\vc{f}^i_{k-1})\\&
+(\tilde{\mat{A}}^i_k\tilde{\mat{J}}^i_k)_{\hv^{-i}_k,a^{-i}_{k-1}}(a^{-i}_{k-1}-m^{-i}_{k-1})+(\vc{d}^i_k)_{\hv^{-i}_k}\\
& = \mat{E}^i_k v+\vc{f}^i_k,
\end{aligned}
\end{equation}
where 
 \begin{equation}
\begin{aligned}
\mat{E}^i_k &=(\mat{G}^{i,-i}_k)_{:,v}+(\mat{G}^{i,-i}_k)_{:,\hv^{-i}_{k-1}} \mat{E}^i_{k-1}\\&+(\tilde{\mat{A}}^i_k\tilde{\mat{J}}^i_k)_{\hv^{-i}_k,a^{-i}_{k-1}}\mathfrak{D}(\mat{L}^{-i}_{k-1})\mat{E}^i_{k-1},\\
\vc{f}^i_k&=((\mat{G}^{i,-i}_k)_{:,\hv^{-i}_{k-1}}-(\tilde{\mat{A}}^i_k\tilde{\mat{J}}^i_k)_{\hv^{-i}_k,a^{-i}_{k-1}}\mathfrak{D}(\mat{L}^{-i}_{k-1}))\vc{f}^i_{k-1}\\&+(\tilde{\mat{A}}^i_k\tilde{\mat{J}}^i_k)_{\hv^{-i}_k,a^{-i}_{k-1}}(a^{-i}_{k-1}-m^{-i}_{k-1})+(\vc{d}^i_k)_{\hv^{-i}_k},
\end{aligned}
\label{EFupdate}
\end{equation}
and
\begin{equation}
\tilde{\mat{J}}^i_k=\tilde{\mat{\Sigma}}^i_{k|k-1}\tilde{\mat{C}}^{i'}_k(\tilde{\mat{C}}^i_k \tilde{\mat{\Sigma}}^i_{k|k-1} \tilde{\mat{C}}^{i'}_k )^{-1},
\label{Jtilde}
\end{equation}
\begin{equation}
\tilde{\mat{\Sigma}}^i_{k|k-1}=\tilde{\mat{A}}^i_{k-1}\tilde{\mat{\Sigma}}^i_{k-1}\tilde{\mat{A}}^{i'}_{k-1}+\tilde{\mat{H}}^i_{k-1} \mathfrak{D}(\mat{Q}^{-i})\tilde{ \mat{H}}^{i'}_{k-1},
\label{condsigmatilde}
\end{equation}
and $\tilde{\mat{\Sigma}}^i_k$ is the covariance matrix of $\tilde{S}^i_k$ conditioned on $\tilde{y}^i_{1:k}$ and is derived from the following recursive update equation
\begin{equation}
\begin{aligned}
&\tilde{\mat{\Sigma}}^i_{k+1}\hspace{-0.1cm}=\hspace{-0.05cm}(\mat{I}-\tilde{\mat{J}}^i_{k+1}\tilde{\mat{C}}^i_{k+1})(\tilde{\mat{A}}^i_k\tilde{\mat{\Sigma}}^i_k\tilde{\mat{A}}^{i'}_k\hspace{-0.1cm}+\hspace{-0.05cm}\tilde{\mat{H}}^i_k \mathfrak{D}(\mat{Q}^{-i}) \tilde{\mat{H}}^{i'}_k)\\
&\tilde{\mat{\Sigma}}^i_1=\mathbb{E}[\tilde{S}^i_1\tilde{S}^{i'}_1]-\mathbb{E}[\tilde{S}^i_1 V'](\mathbb{E}[VV'])^{-1}\mathbb{E}[\tilde{S}^i_1 V']'\\
&=\left[
\begin{array}{cccc} \mat{\Sigma} & \mat{0} \\
\mat{0} & \mat{0} 
\end{array}\right]-
\left[\hspace{-0.1cm}\begin{array}{c} \mat{\Sigma} \\
\mat{0} 
\end{array}\hspace{-0.1cm}\right]
\hspace{-0.05cm}\mat{\Sigma}^{-1}\hspace{-0.05cm}\left[\hspace{-0.1cm}\begin{array}{cccc} \mat{\Sigma}' & \mat{0}
\end{array}\hspace{-0.1cm}\right].
\end{aligned}
\label{Sigmatilde}
\end{equation}
\end{itemize}
\end{proof}

By using Theorem~\ref{th1}, one can form the summary  $S(h^i_t)$ for the specific (linear) strategies of all players as mentioned in the theorem. This will enable us to form an LQG model for player $i$ and prove the optimality of linear strategy for her, given others play linear strategies.

\subsection{Linear Quadratic Gaussian (LQG) model from player $i$'s perspective}
Part $(a)$ of Theorem \ref{th1} implies that $S^i_t$ is a jointly Gaussian random vector conditioned on player $i$'s observation till time $t$, $\forall i \in \cN, t \in \cT$. This implies that the beliefs over $V$ are jointly Gaussian and so players need only keep track of their belief's mean (estimation) and covariance matrices. Furthermore, this theorem implies that a player's belief over other players beliefs is also Gaussian and hence, players need to keep track of their estimation on other players' estimations, i.e., $\hhv$. The important point of Theorem~\ref{th1} is the statement that the estimation of players on others' estimations is a linear function of their own estimation and hence, in order to keep track of the estimation over other players' estimations, a player only needs to keep track of her own estimation over $V$. Therefore,  $\hv^i_t$ is a sufficient statistic for player $i$'s private observations till time $t$.

On the other hand, in the proof of Theorem \ref{th1}, there are three quantities, $\mat{\Sigma}^i_t$, $\mat{E}^i_t$ and $f^i_t$, that are updated recursively as a function of previous strategies and actions. This means they can not be calculated off-line like the covariance matrix in the classic LQG stochastic control problem~\cite[Ch.7]{KuVa86}.
A way to resolve this issue is to consider them as the public summary of $h^i_t$, i.e., $P(h^i_t)$, thus leading to strategies of the form $\psi^i_t(\cdot|\hv^i_t,\mat{\Sigma}^i_t,\mat{E}^i_t,\vc{f}^i_t)=\gamma^i_t(\cdot|\hv^i_t)$.
In particular, we will now show that linear strategies of the form $\gamma^i_t(\cdot|\hv^i_t)=\delta(a^i_t-\mat{L}^i_t\hv^i_t-m^i_t)$ are PBE of the game by showing that if every player $j\in -i$ is playing according to $(\gamma^{-i}_k)_{k\leq t}$ and player $i$ is playing according to $(\gamma^i_{k})_{k\leq t-1}$, then player $i$ faces a standard LQG control model from $t$ onwards. By using the results from \cite[Ch.7]{KuVa86}, we can conclude that player $i$'s optimal strategy is linear in $\hv^i_t$. This is summarized in the following theorem.

\begin{theorem}
	For any $t \in \cT$, if all players $-i$ play according to the strategy $\gamma^{-i}_t(a^{-i}_t|\hv^{-i}_t)=\delta(a^{-i}_t-\mathfrak{D}(\mat{L}^{-i}_t)\hv^{-i}_t-m^{-i}_t)$ and for $k < t$, the strategies of players are linear in $\hv_k$, player $i$ faces an MDP with state $(\hv^i_t,\mat{\Sigma}^i_t,\mat{E}^i_t,\vc{f}^i_t)$.
The reward-to-go functions are updated backwards according to
\vspace{0.1cm}
		\begin{equation}
\begin{aligned}
&J^i_t(\hv^i_t,\mat{\Sigma}^i_t,\mat{E}^i_t,\vc{f}^i_t)=\max_{\tilde{\gamma}^i_t(\cdot|\hv^i_t)}\mathbb{E}^{\gamma^{-i}_t,\tilde{\gamma}^i_t(\cdot|\hv^i_t)}[R^i(V,A_t)\\
 &+J^i_{t+1}(\chv^i_{t+1},\mat{\Sigma}^i_{t+1},\mat{E}^i_{t+1},\vc{f}^i_{t+1})|\hv^i_t,\mat{\Sigma}^i_t,\mat{E}^i_t,\vc{f}^i_t]
\end{aligned}
\label{backV}
		\end{equation}
		and
		\vspace{0.1cm}
		\begin{equation}
	\begin{aligned}
	&\gamma^i_t(\cdot|\hv^i_t)=\arg \max_{\tilde{\gamma}^i_t(\cdot|\hv^i_t)}\mathbb{E}^{\gamma^{-i}_t,\tilde{\gamma}^i_t(\cdot|\hv^i_t)}[R^i(V,A_t)\\
&+J^i_{t+1}(\chv^i_{t+1},\mat{\Sigma}^i_{t+1},\mat{E}^i_{t+1},\vc{f}^i_{t+1})|\hv^i_t,\mat{\Sigma}^i_t,\mat{E}^i_t,\vc{f}^i_t],
	\end{aligned}
	\label{backGamma}
		\end{equation}
where 	$\chv^i_{t+1},\mat{\Sigma}^i_{t+1},\mat{E}^i_{t+1},\vc{f}^i_{t+1}$ are generated from $\chv^i_t,\mat{\Sigma}^i_t,\mat{E}^i_t,\vc{f}^i_t$ using $\gamma_t$.

Further, it is optimal for player $i$ to play according to $\gamma^i_t(\cdot|\hv^i_t)=\delta(a^i_t-\mat{L}^i_t\hv^i_t-m^i_t)$.
\end{theorem}
	
\begin{proof}
	By using the results from Theorem \ref{th1}, given the strategy profile $\gamma_t$, $(\hv^i_t,\mat{\Sigma}^i_t,\mat{E}^i_t,\vc{f}^i_t)$ forms a Markov chain. Notice that $\chv^i_{t+1},\mat{\Sigma}^i_{t+1},\mat{E}^i_{t+1},\vc{f}^i_{t+1}$ are updated by $\gamma_t$ which is linear and therefore, all  results from Theorem \ref{th1} hold. Further, the expected reward $\mathbb{E}[R^i(V,A_t)|\hv^i_t,\mat{\Sigma}^i_t,\mat{E}^i_t,\vc{f}^i_t]$ can be written as $quad(\tilde{\bB}^i_t;\hv^i_t)+\rho^i_t$~\cite{VaAn16a} for some appropriately defined matrix $\tilde{\bB}^i_t$ and function $\rho^i_t$. Hence, the expected reward is measurable with respect to $(\hv^i_t,P(h^i_t))$. We conclude that player $i$ faces an MDP. Further, since at each time $t \in \cT$, player $i$ faces and MDP with quadratic reward with respect to $\hv^i_t$, she faces an LQG.  We refer to \cite[Ch.7]{KuVa86} to conclude that it is optimal for her to play according to $\gamma^i_t(\cdot|\hv^i_t)=\delta(a^i_t-\mat{L}^i_t\hv^i_t-m^i_t)$, where $\mat{L}^i_t$ and $m^i_t$ are functions of the public summary, $\mat{\Sigma}^i_t,\mat{E}^i_t,\vc{f}^i_t$, and quantities $\mat{L}^{-i}_t$ and $m^{-i}_t$.	
	\end{proof}

	\section{Constructing Structured PBE}\label{section5}
The construction of the mentioned structured PBE is summarized in the following backward/forward sequential decomposition algorithm.

	\subsection{Backward Programming}
	For every $i \in \cN$,
	\begin{itemize}[noitemsep,wide=0pt, leftmargin=0pt]
	\item Set $t=T+1$.
	 \item For every $(\hv^i_t,\mat{\Sigma}^i_t,\mat{E}^i_t,\vc{f}^i_t)$, set $J^i_t(\hv^i_t,\mat{\Sigma}^i_t,\mat{E}_t,\vc{f}_t)=0$.
	\item Set $t=t-1$.
	\item For every $(\hv^i_t,\mat{\Sigma}^i_t,\mat{E}^i_t,\vc{f}^i_t)$, set the value of $J^i_t(\hv^i_t,\mat{\Sigma}^i_t,\mat{E}^i_t,\vc{f}^i_t)$ and $\gamma^i_t(\cdot|\hv^i_t)$ according to \eqref{backV} and \eqref{backGamma}. Set $\psi^i_t(\cdot|\hv^i_t,\mat{\Sigma}^i_t,\mat{E}^i_t,\vc{f}^i_t)=\gamma^i_t(\cdot|\hv^i_t)$.
	\item If $t>1$, go to step 3, otherwise stop.
	\end{itemize}
\subsection{Forward Programming}
For every $i \in \cN$,
\begin{itemize}[noitemsep,wide=0pt, leftmargin=0pt]
	\item Set $t=1$.
	\item According to $x^i_t$, set $\hv^i_t$, $\mat{\Sigma}^i_t$, $\mat{E}^i_t$ and $\vc{f}^i_t$ according to \eqref{basisV}, \eqref{Sigma} and \eqref{basisEF}, respectively.
	\item Set $\gamma^i_t(\cdot|\hv^i_t)=\psi^i_t(\cdot|\hv^i_t, \mat{\Sigma}^i_t, \mat{E}^i_t, \vc{f}^i_t)$.
       \item Set $t=t+1$.
	\item  Update $\hv^i_t$, $\mat{\Sigma}^i_t$, $\mat{E}^i_t$, $\vc{f}^i_t$ according to \eqref{vhupdate}, \eqref{Sigma}, \eqref{EFupdate}.
	\item If $t<T$, go to step 3, otherwise, stop.
\end{itemize}
\section{Conclusion} \label{section6}
In this paper, we studied a dynamic LQG game with asymmetric information and dependent types. We considered linear strategies for players and by using conditional independence of types and Kalman filter results, we proved that beliefs of players are Gaussian. Furthermore, each player's estimate over other players' estimates are public functions of her own estimates. This fact enabled us to construct a summary of players' histories at each time and develop an LQG model from the perspective of each player.  We thus characterized PBE with linear strategies through a sequential backward/forward algorithm.

Future work for this model includes investigation of conditions under which we can have steady state equilibria.
In addition, we are planning to investigate conditions on the problem primitives under which the described sequential decomposition algorithm is guaranteed to have solutions.

\bibliographystyle{IEEEtran}
\input{root_bib.bbl}

\end{document}